\documentclass[9pt]{entcs} \usepackage{prentcsmacro}
\usepackage{graphicx,mathrsfs,amsmath}
\usepackage{tabularx,multirow,booktabs}
\usepackage[table]{xcolor}

\definecolor{cellcolor}{rgb}{0.94, 0.92, 0.84}

\newcommand{\CC}{\mathscr{C}}
\newcommand{\AC}{\mathscr{A}}
\newcommand{\ST}{\mathscr{S}}

\def\lastname{Lucci, Nasini, Sever\'in}
\begin{document}
\begin{frontmatter}
  \title{A Branch and Price Algorithm for List Coloring Problem}
	  \author{Mauro Lucci\thanksref{ALL}\thanksref{myemail}}
  \address{Depto. de Cs. de la Computaci\'on, FCEIA, Universidad Nacional de Rosario, Argentina\\
    CONICET, Argentina} 
    \author{Graciela Nasini\thanksref{coemail}}
  \address{Depto. de Matem\'atica, FCEIA, Universidad Nacional de Rosario, Argentina\\
    CONICET, Argentina} 
    \author{Daniel Sever\'in\thanksref{co2email}}
  \address{Depto. de Matem\'atica, FCEIA, Universidad Nacional de Rosario, Argentina\\
    CONICET, Argentina} 
  \thanks[ALL]{This work is supported by grants ANPCyT PICT-2016-0410 and PID-UNR ING538.}
	\thanks[myemail]{Email: \href{mailto:mlucci@fceia.unr.edu.ar} {\texttt{\normalshape mlucci@fceia.unr.edu.ar}}}
	\thanks[coemail]{Email: \href{mailto:nasini@fceia.unr.edu.ar} {\texttt{\normalshape nasini@fceia.unr.edu.ar}}}
	\thanks[co2email]{Email: \href{mailto:daniel@fceia.unr.edu.ar} {\texttt{\normalshape daniel@fceia.unr.edu.ar}}}

\begin{abstract}
Coloring problems in graphs have been used to model a wide range of real applications.
In particular, the List Coloring Problem generalizes the well-known Graph Coloring Problem for which many exact algorithms have been developed. 
In this work, we present a Branch-and-Price algorithm for the weighted version of the List Coloring Problem, based on the one developed by Mehrotra and Trick (1996) for the Graph Coloring Problem.
This version considers non-negative weights associated to each color and it is required to assign a color to each vertex from predetermined lists in such
a way the sum of weights of the assigned colors is minimum.
Computational experiments show the good performance of our approach, being able to comfortably solve instances whose graphs have up to seventy vertices.
These experiences also bring out that the hardness of the instances of the List Coloring Problem does not seem to depend only on quantitative parameters such as
the size of the graph, its density, and the size of list of colors, but also on the distribution of colors present in the lists.
\end{abstract}
\begin{keyword}
List Coloring, Branch and Price, Weighted Problem. 
\end{keyword}
\end{frontmatter}

\section{Introduction} \label{intro}

The Graph Coloring Problem (GCP) models a wide range of planning problems such as timetabling, scheduling, electronic bandwidth allocation and sequencing.
Some applications require to impose additional constraints to the GCP, giving rise to known variants such as Equitable Coloring,
Precoloring Extension, $(\gamma,\mu)$-coloring and List Coloring, see e.g.~\cite{kubale2004,willy2009}.
Actually, List Coloring generalizes GCP, Precoloring Extension and $(\gamma,\mu)$-coloring and has several specific applications such as channel allocation in
wireless networks \cite{wang2005}.
In many practical situations, colors have different weights (or costs) and it is required to find the coloring of minimum weight instead of minimum cardinality.
Particularly, in \cite{ejs2018}, the design of workdays of drivers in a public transport company is modeled as a Minimum Weighted List Coloring Problem (MWLCP). 

Let $G = (V,E)$ be an undirected simple graph and $\CC$ be a set of colors.
A \emph{coloring of} $G$ is a function $f : V \rightarrow \CC$ such that $f(u) \neq f(v)$ for every edge $(u,v)$ of $G$.
Given a coloring $f$ of $G$, the \emph{class color of} $j\in \CC$, denoted by $f^{-1}(j)$, is the set of vertices $v$ colored by $j$, i.e.~such that $f(v)=j$.
Clearly, each class color is a stable set of $G$ and, therefore, any coloring can be thought as a partition of vertices into stable sets.
The \emph{active colors} of a coloring $f$, denoted by $\AC$, are those ones assigned to some vertex,
i.e.~$\AC \doteq \{ j \in \CC : f^{-1}(j) \neq \emptyset\}$.
The GCP consists of finding a coloring of $G$ with minimum cardinality of $\AC$.
This minimum is called \emph{the chromatic number of} $G$, denoted by $\chi(G)$, and it is known that obtaining this number is $\mathcal{NP}$-hard for general graphs.

Despite its hardness, the need of obtaining concrete solutions to numerous applications motivated the development of several exact algorithms for GCP:
combinatorial branch-and-bound such as DSATUR \cite{brelaz1979,sansegundo2012},
branch-and-cut BC-COL \cite{imendez2006} and branch-and-price LPCOLOR \cite{trick1996,held2011,toth2011}.

In the case of BC-COL, an Integer Linear Programming (ILP) compact formulation (\emph{GCP-CF} from now on) based on classic vertex-color assignment variables is used.
Instead, LPCOLOR is based on set-covering ILP formulation (\emph{GCP-SC} from now on), where each variable represents a stable set of $G$.
As the number of variables is usually exponential in the size of $G$, the resolution of the linear relaxation of GCP-SC is addressed by a column generation procedure.\\

In List Coloring, each vertex $v\in V$ has a preassigned \emph{list} $L(v)\subset \CC$ defining those colors that can be assigned to $v$.
Formally, a \emph{list coloring of} $G$ is a coloring $f$ of $G$ with the additional condition that $f(v) \in L(v)$ for all $v\in V$.
Note that classic colorings of $G$ are list colorings for which $L(v)=\CC$ for all $v\in V$.

Given a vector $w\in \mathbb{Z}_+^{\CC}$, the MWLCP consists of finding a list coloring of $G$ such that the sum of weights of the active colors,
i.e.~$\sum_{j \in \AC} w_j$, is minimum.
As colorings of $G$ can be thought as particular cases of list colorings, MWLCP can be seen as a generalization of a \emph{weighted version} of GCP.
However, this weighted version can be trivially reduced to GCP since one can pick the cheapest $\chi(G)$ colors from $\CC$ to obtain the coloring of minimum weight.

Clearly, MWLCP is $\mathcal{NP}$-hard in general graphs since it generalizes GCP.
However, some known results reveal that MWLCP is indeed \emph{harder} than GCP.

As a first remark, unlike GCP, which is feasible whenever $|\CC| \geq \chi(G)$, MWLCP can be infeasible even if $|L(v)|\geq \chi(G)$ for all $v$.
A well-known example is presented in Figure \ref{fig:infeas} where the available colors for each vertex are enclosed in braces.\\
\begin{figure}
	\centering
		\includegraphics[scale=0.2]{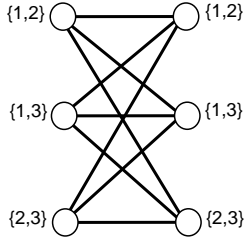}
		\caption{An infeasible instance} \label{fig:infeas}
\end{figure}

Actually, asking whether $G$ has a list coloring is $\mathcal{NP}$-complete even when restricted to instances where the size of lists is at most 3 and $G$ is a cograph or complete bipartite \cite{golovach2017}. In contrast, GCP can be solved on cographs and bipartites in linear time. %More classes are shown in Table 1 from \cite{willy2009}.

Another example where these problems differ in difficulty is when graphs are parameterized by treewidth.
For a given graph $G$ with fixed treewidth $t$, there is a dynamic programming algorithm that finds $\chi(G)$ in polynomial time, while asking if $G$ has a list
coloring is W[1]-Hard \cite{Computcomplex}.\\

Formulation GCP-CP (resp.~GCP-SC) mentioned above can be extended to the MWLCP in such a way that they coincide when applied to instances of GCP.
In this work, we present these \emph{extended} ILP formulations, respectively called MWLCP-CP and MWLCP-SC, and we develop a Branch-and-Price algorithm to solve MWLCP via
MWLCP-SC, by following some ideas provided in \cite{trick1996}. 
However, we have to take into account some additional issues that arise from the difference between GCP and MWLCP.
In GCP every vertex can be colored with any element of $\CC$ and all colors have the same weight, thus making the colors \emph{indistinguishable} from each other.
Clearly, this property is not met for MWLCP forcing us to consider one for each stable set and each color in MWLCP-SC.
In spite of that, we designed a mechanism to take advantage of those cases where the indistinguishability between colors is partially met.
Feasibility is another issue, as mentioned before.

In Section \ref{ilps}, we formally present the ILP formulations MWLCP-CP and MWLCP-SC. 
In Sections \ref{solvelprel} and \ref{brrule}, we explain the main components that define our Branch-and-Price algorithm: 
the way LP relaxations are solved, how columns are generated and how the node branching is performed.
In Section \ref{compu}, computational experiments are conducted to evaluate the performance of our algorithm over random instances of
different sizes and types. Finally, in Section \ref{conclu}, some conclusions are drawn.

\section{ILP formulations for MWLCP} \label{ilps}

Consider an instance of the MWLCP, given by a graph $G = (V,E)$, a set of colors $\CC$ with weights $w_j\in \mathbb{Z}_+$ for each $j \in \CC$ and lists
$L(v)\subset \CC$ for each vertex $v$. W.l.o.g we can assume that $G$ is connected and $\CC = \bigcup_{v \in V} L(v)$. 

For each color $j$, let $V_j \doteq \{ v \in V : j \in L(v) \}$, $G^j$ be the subgraph of $G$ induced by $V_j$, $I_j$ be the set of isolated vertices in $G^j$ and
$\ST(G^j)$ be the set of maximal stable sets of $G^j$.

The first formulation is based on GCP-CF \cite{imendez2006}. We have a binary variable $x_{vj}$, for each $v \in V$ and $j \in L(v)$, such that its value is 1 if and only if $f(v) = j$, and binary variables $y_j$ such that $y_j=1$ if $j$ is an active color in $f$.
Then, the MWLCP can be formulated as follows:
\begin{align}
\textrm{(MWLCP-CF)}~~ & \min \sum_{j \in \CC} w_j y_j & & \notag \\
  & s.t. & & \notag \\
  & \sum_{j \in L(v)} x_{vj} = 1 & \forall~v \in V & \label{CFRESTR1}\\ 
  & x_{uj} + x_{vj} \leq y_j & \forall~j \in \CC,~(u,v) \in E(G^j) & \label{CFRESTR2} \\
	& x_{vj} \leq y_j & \forall~j \in \CC,~v \in I_j & \label{CFRESTR3} \\
	& x_{vj} \in \{0, 1\} & \forall~v \in V,~j \in L(v) & \notag \\
	& y_j \in \{0, 1\} & \forall~j \in \CC \notag
\end{align}
Constraints \eqref{CFRESTR1} ensure that each vertex is assigned a color, \eqref{CFRESTR2} impose that no pair of adjacent vertices receive the same color (when that color is used), and \eqref{CFRESTR3} control the activation of $y_j$ when $v$ is an isolated vertex in $G^j$.

\medskip

The second ILP model corresponds to a formulation of colorings as covering of vertices by stable sets, based on GCP-SC.
Here, we have a binary variable $x_S^j$ for each $j\in\CC$ and each stable set $S\in \ST(G^j)$.
Observe that, for any color $j$, its color class is a stable set from $G^j$ and, unlike GCP, no more than one maximal stable set of $G^j$ can be used by the covering.
The following ILP model arises naturally:
\begin{align}
~~ & \min \sum_{j \in \CC} w_j \sum_{S \in \ST(G^j)}  x^j_S & & \notag \\
  & s.t. & & \notag \\
  & \sum_{j \in L(v)} \sum_{S \in \ST(G^j) : v \in S} x^j_S \geq 1 & \forall~v \in V & \label{SCRESTR1X}\\ 
  & \sum_{S \in \ST(G^j)} x^j_S \leq 1 & \forall~j \in \CC & \label{SCRESTR2X} \\
	& x^j_S \in \{0, 1\} & \forall~j \in  \CC,~S \in \ST(G^j) & \notag
\end{align}
Constraints \eqref{SCRESTR1X} say that, for each vertex $v$, at least one stable set containing $v$ is selected, while
\eqref{SCRESTR2X} allow to take at most one stable set for each color.

However, this formulation does not properly generalize GCP-SC in the sense that an instance of GCP gives rise to a model with many more
variables than GCP-SC. As we have said in the introduction, GCP-SC takes advantage of the indistinguishability between colors.
In order to overcome this drawback, we say that two colors $j,k \in \CC$ are indistinguishable if $w_j=w_k$ and $G^j=G^k$.
Let $\CC^k = \{ j \in \CC : j, k~\textrm{are indistinguishable} \}$ and $\{ \CC^k : k \in K \}$ be a partition of $\CC$ where $K\subset \CC$.
In other words, each $k \in K$ is a color that \emph{represents} those ones from $\CC^k$.

Now, we can take advantage of the symmetry between those colors and allow the existence of at most $|\CC^k|$ stable sets of $G^k$ instead of a single one.
Moreover, if $|G^k|\leq |\CC^k|$, such constraint can be removed since there is an optimal solution verifying this condition.
Now, the MWLCP can be re-formulated:
\begin{align}
\textrm{(MWLCP-SC)}~~ & \min \sum_{k \in K} w_k \sum_{S \in \ST(G^k)}  x^k_S & & \notag \\
  & s.t. & & \notag \\
  & \sum_{k \in K} \sum_{S \in \ST(G^k) : v \in S} x^k_S \geq 1 & \forall~v \in V & \label{SCRESTR1}\\ 
  & \sum_{S \in \ST(G^k)} x^k_S \leq |\CC^k| & \forall~k \in K : |G^k|\geq |\CC^k|+1 & \label{SCRESTR2} \\
	& x^k_S \in \{0, 1\} & \forall~k \in K,~S \in \ST(G^k) & \notag
\end{align}
Therefore, for instances corresponding to GCP, we have $K=\{1\}$ and $\CC^1 = \CC$.
If, in addition, $|\CC| \geq |V(G)|$, constraints \eqref{SCRESTR2} are no longer needed and our last formulation turns out to be the same as GCP-SC.

\section{Solving LP relaxations of MWLCP-SC} \label{solvelprel}

Given an instance $(G,\CC,w, L)$ of MWLCP, let $X \doteq \{(S,k) : S \in \ST(G^k),~ k \in K \}$.
The variables of the linear relaxation of MWLCP-SC are $x^k_S$ for $(S,k)\in X$ and its constraints can be written as $Ax\geq 1$, $Bx\geq -b$ and $x\geq 0$ where, for
every variable $x^k_S$, the corresponding column in $A$ is the characteristic vector of $S$ and the columns in $B$ have a non-zero entry, $-1$, in the row corresponding
to $k\in K$. Moreover, for each $k\in K$, $b_k=|\CC^k|$.
Note that constraints of the form $x^k_S \leq 1$ are not necessary (non-negative weights ensure that at least one optimal solution verifies these conditions).
%For the sake of simplicity, we assume that $|G^k| \leq |\CC^k|$ for all $k\in K$ so a constraint \eqref{SCRESTR2} for each $k$ is present in the relaxation. (NO ES NECESARIO, PORQUE DECIMOS "AT MOST" EN LA MATRIX B')

For any $\hat X \subset X$, let $x(\hat X)$ be the vector of variables $x^k_S$ restricted to $(S,k) \in \hat X$ and $LP(\hat X)$ be the linear relaxation of
MWLCP-SC restricted to $x(\hat X)$.
We first give some results concerning the integrability of solutions of the LP relaxation.

\begin{lemma} \label{lemita}
Let $x^*$ be an optimal solution of $LP(X)$ and let $\tilde X = \{(S,k) \in X : |S|\geq 2 \}$.
If $x^*(\tilde X)$ is integral then $LP(X)$ has an optimal integer solution.
\end{lemma}
\begin{proof}
Consider the partition $\{ X^0, X^{+} \}$ of $X$ where $X^0 \doteq \{ (S, k) : x^{*k}_S = 0 \}$ and $X^{+} \doteq X \setminus X^0$.
Clearly, $x^*(X^{+})$ is an optimal solution for $LP(X^{+})$.
Since $x^*(\tilde X)$ is integral, variables from $\tilde X \cap X^{+}$ are set to 1.
Then, by replacing occurrences of those variables by 1 in $LP(X^{+})$, we obtain a linear program $LP'$ over variables
$X^{+} \setminus \tilde X$, such that $x^*(X^{+} \setminus \tilde X)$ is an optimal solution for $LP'$.
Observe that the stable sets represented by variables in $X^{+} \setminus \tilde X$ are singletons.

Let $A'$ and $B'$ be the submatrices of $A$ and $B$ corresponding to constraints \eqref{SCRESTR1} and \eqref{SCRESTR2} in $LP'$, respectively. 
Since columns of $A'$ are characteristic vectors of singleton stable sets, they have one positive entry equal to 1.
Similarly, columns of $B'$ have at most one negative entry equal to -1, thus implying that the matrix of coefficients of $LP'$ is totally unimodular.
Therefore, there exist integer optimal solutions for $LP'$ and also $LP(X)$.
\end{proof}

Note that when every $G^j$ is complete, stable sets are singletons and $\tilde X = \emptyset$ in the statement of the previous lemma. Then, we obtain:

\begin{corollary} \label{completos}
If, for all $j \in \CC$, $G^j$ is a complete graph, then $LP(X)$ has an integral optimal solution.
\end{corollary}

Actually, the result given in the previous corollary can be also derived from a reduction of MWLCP to 
the Minimum Weighted Perfect Matching Problem on Bipartite Graphs, described below.

Consider an instance $(G,\CC,w, L)$ of MWLCP such that $G^j$ is a complete graph for all $j\in \CC$.
Observe that each color can be assigned to at most one vertex. Then, we can assume that $|V(G)| \leq \CC$ (otherwise, MWLCP would be infeasible). Let $Z$ be a set of dummy elements such that $|Z| = |\CC| - |V(G)|$.

Now, consider the bipartite graph $\tilde G$  with vertex partition defined by $V(G)\cup Z$ and $\CC$, and such that for all $v\in V(G)$ and
$j\in \CC$, $(v,j)$ is an edge of $\tilde G$ if and only if $j\in L(v)$.
The weight of $(v,j)$ is $w_j$.
We also consider edges $(z,j)$ for all $z\in Z$ and $j\in \CC$ with weight zero.
Clearly, if $f$ is a list coloring of $G$, there exists a perfect matching of $\tilde G$ containing edges $(v,f(v))$ for all $v\in V(G)$ and edges $(z,j(z))$ for all $z\in Z$ and some $j(z)$ taken from the non-active colors in $f$.
The weight of this perfect matching coincides with the weight of the list coloring. 
Conversely, for any perfect matching $M$ of $\tilde G$, edges $(v,j)\in M$ for every $v\in V(G)$ define a list coloring of $G$ with the same weight that $M$. Moreover, $\tilde G$ has a perfect matching if and only if MWLCP is feasible.

In practice, this problem can be addressed much faster via the Hungarian Algorithm than by solving the LP relaxation.

\subsection{Column generation}

Let $\hat X\subset X$, $x^*(\hat X)$ be an optimal solution of $LP(\hat X)$, and $(\pi^*, \gamma^*)$ be the optimal dual solution of $LP(\hat X)$,
where $\pi$ and $\gamma$ are the dual variables corresponding to constraints \eqref{SCRESTR1} and \eqref{SCRESTR2}, respectively
(for those $k \in K$ such that $|G^k| \leq |\CC^k|$, we assume $\gamma^*_k = 0$).
Note that, for any $(S,k)\in X$, the reduced cost of $x^k_S$ is $c^k_S \doteq \sum_{v\in S} \pi^*_v - (w_k + \gamma^*_k)$.

We recall that the $x^*$ obtained from $x^*(\hat X)$ by padding zeros, i.e.~$x^*(X \setminus \hat X) = {\bf 0}$, is feasible for
$LP(X)$ but not necessarily optimal.
In order to know that, we have to decide if there exists $(S,k)\in X$ such that $c^k_S > 0$ or, equivalently, to solve the following
auxiliary problem:
\begin{center}
(Aux)~~~ $\exists$~$(S,k) \in X$ \emph{such that} $\sum_{v\in S} \pi^*_v > w_k + \gamma^*_k$?
\end{center}
Clearly, (Aux) is equivalent to solve, for each $k\in K$, the \emph{Maximum Weighted Stable Set Problem} (MWSSP) on $G^k$
with weight $\pi^*_v$ for each $v\in V(G^k)$. 

Despite the MWSSP is $\mathcal{NP}$-hard, the enumerative routine given in \cite{held2011} shows to be reasonably fast when it is
called by the branch-and-price algorithm.
In addition, it is not always necessary to solve the MWSSP to optimality.
One can stop the optimization as soon as a set $S\in \ST(G^k)$ with weight greater that the threshold $T^k \doteq w_k + \gamma^*_k$ is found.

Moreover, if we have $G^k=G^\ell$ and $T^k > T^\ell$ for some $k \neq \ell \in K$, there is no need to run the enumerative routine
for $G^\ell$ since the stable set found for $G^k$ can be reused for $G^\ell$.

Although finding one column is enough for the column generation process, preliminary experiments have shown that incorporating many
columns at the same time decreases the number of iterations of the process and, thus, the overall time.
In our implementation, we search for (if exists) an entering variable per $k \in K$.

Observe that the column generation process needs to start with a set $\hat X\subset X$ for which $LP(\hat X)$ is feasible.
We recall that, unlike GCP, to know if an instance of MWLCP is feasible is an $\mathcal{NP}$-complete problem and the best algorithm that finds a list coloring, whenever it exists, is $O(2^n)n^{O(1)}$ \cite{husfeldt2009}.
In the following section we propose an initialization procedure.

\subsection{Initializing the linear relaxations}   \label{initrel}

Our approach consists in extending the solution space by creating a dummy color per vertex.
Given an instance $(G,\CC,w,L)$ of MWLCP we generate an extended instance $(G,\CC',w',L')$ where $\CC'= \CC \cup \{d_v : v \in V(G)\}$, $w'_j = w_j$ for all $j \in \CC$ and $w'_{d_v} = M$ for all $v\in V(G)$, where $M$ is a \emph{big number} (e.g.~$M=1+ \sum_{j\in \CC} w_j$). In addition, $L'(v) = L(v) \cup \{d_v\}$ for all $v\in V(G)$.

Now, the linear relaxation obtained from the instance $(G,\CC',w',L')$ is trivially feasible since
$x_{\{v\}}^{d_v}=1$ for all $v\in V(G)$ is a feasible solution which allows us to start our column generation process.
That is, $\hat X = \{(\{v\}, d_v): v \in V(G)\}$.

It is expected that dummy colors will no longer be used as optimization proceeds.
In fact, if some dummy color is still active in the optimal solution of the relaxation associated to the extended instance $(G,\CC',w',L')$,
then $(G,\CC,w,L)$ is infeasible.\\

Preliminary experiments have revealed that starting with dummy colors is better than a feasible initial coloring, when the latter can be generated (there are heuristics
algorithms that deliver list colorings, such as $k$-Greedy-List \cite{molloy1997}).
We speculate the cause of this may be that the initial columns yielded by such colorings are of poor quality, thus making the LP solver to perform more iterations.\\

Since our column generation process is embeeded in a Branch-and-Price framework, the subproblems we have to solve for each node of the
B\&B tree should correspond to new instances of the MWLCP.
To achieve that purpose, we need to devise a \emph{robust branching rule}. In the next section we present such a rule,
adapting the one first proposed in \cite{zykov1949} and used in \cite{trick1996} for the GCP.

\section{A robust branching rule}  \label{brrule}

The idea behind the rule given in \cite{trick1996} is to pick two non-adjacent vertices $u$ and $v$ and divide the space of solutions between those ones
satisfying $f(u) = f(v)$ (both share the same color) and those others satisfying $f(u) \neq f(v)$.
Finding the optimal coloring in the latter case is equivalent to solve GCP on a graph obtained by adding edge $(u,v)$ to $G$. In the former case, one has to solve GCP on a graph obtained from $G$ by collapsing vertices $u$ and $v$ into a single one (that is, connecting $u$ to
every vertex from $N_G(u) \cup N_G(v)$ and removing $v$, where $N_G(x)$ is the open neighborhood of $x$ in $G$).

We follow the same strategy, with the additional condition that vertices $u$ and $v$ must satisfy $L(u)\cap L(v)\neq \emptyset$.
Indeed, finding the optimal list coloring $f$ of $(G, \CC, w, L)$ with $f(u) \neq f(v)$ is equivalent to solve MWLCP for the instance
$(G_1, \CC, w, L)$ where $G_1 = (V(G), E(G) \cup \{u,v\})$, whereas finding the optimal $f$ with $f(u) = f(v)$ is equivalent to solve
MWLCP for the instance $(G_2, \CC, w, L_2)$ where $G_2 = G \setminus \{v\}$ and $u$ is connected to every vertex
from $N_G(v)$, $L_2(z) = L(z)$ for all $z \in V(G_2) \setminus \{u\}$ and $L_2(u) = L(u) \cap L(v)$.  

Note that, in a finite number of applications of the previous rule, we get instances where all graphs $G^k$ are completes.
At that point, the optimal solutions of their LP relaxations are integers by Corollary \ref{completos}. 

In addition, due to the intersection performed to obtain $L_2(u)$, it is likely to reach subproblems whose instances have a vertex
with a singleton list. Since such vertex is forced to be colored with the unique available color for it, the instance can be
preprocessed according to the following result:
\begin{lemma} \label{uncolor}
Let $(G, \CC, w, L)$ be an instance of the MWLCP such that $L(u) = \{j\}$ for some $u \in V(G)$ and $j \in \CC$.
Solving $(G, \CC, w, L)$ is equivalent to solve MWLCP for an instance $(G', \CC, w, L')$ where $G' = G \setminus \{u\}$,
$L'(z) = L(z) \setminus \{j\}$ for all $z \in N(u)$ and $L'(z) = L(z)$ for all $z \in V(G') \setminus N(u)$.
\end{lemma}
\begin{proof}
If $f'$ is an optimal list coloring of $(G', \CC, w, L)$ then $f(u) = j$, $f(z) = f'(z)$ for all $z \in V(G')$ is an optimal list coloring
of $(G, \CC, w, L)$. Instead, if $(G', \CC, w, L)$ is infeasible, $(G, \CC, w, L)$ is also infeasible.
\end{proof}

It only remains to specify a criterion for selecting vertices $u$ and $v$, which we address below.

\subsection{Branching variable selection strategy} \label{criterion}

In \cite{trick1996}, a simple but efficient rule is proposed.
We explain it briefly in terms of our formulation for the case $K=\{1\}$ (superscripts are omitted).
Assume that $x^*$ is a non-integer optimal solution of the LP relaxation.
The authors first search for the \emph{most fractional variable}, i.e.~an $x^*_{S_1}$ that minimizes $|x^*_{S_1} - \frac{1}{2}|$, and pick $u \in S_1$.
Since $x^*_{S_1} \notin \{0,1\}$ and \eqref{SCRESTR1} holds for $u$, there exists another $S_2$ such that $x^*_{S_2} > 0$ and $u \in S_2$.
Then, they search for the first vertex $v$ such that $v \in S_1 \Delta S_2$ and, due to the fact that stable sets are different from each other,
such $v$ exists and it is not adjacent to $u$.

In our case, we search for the most fractional $x^{*k}_{S_1}$ satisfying $|S_1| \geq 2$ and we pick $u \in S_1$.
Lemma \ref{lemita} guarantees the existence of such stable set (otherwise, the solution would be integral).
Now we need a pair $(S_2, \ell)$ different from $(S_1, k)$ such that $u\in S_2$ and $x_{S_2}^{*\ell}>0$.
If such $(S_2, \ell)$ exists, we pick $v \in S_2 \setminus S_1$. Otherwise, we pick $v \in S_1\setminus\{u\}$.
In this way, we can assert that $u$ and $v$ are non-adjacent and satisfy the pre-condition $L(u)\cap L(v) \neq \emptyset$ since they must have colors $k$ or $\ell$
in their lists.

\section{Computational results} \label{compu}

This section is devoted to analyze the performance of the proposed Branch-and-Price algorithm
over randomly generated instances, with number of vertices $n\in \{50,60,70\}$ and edge probability $p\in \{0.25,0.50,0.75\}$.

Cardinality of $\CC$ and distribution of colors in the lists are ruled by two parameters $c\in \{0.5,1.0,1.5\}$ and $q\in \{0.25,0.50,0.75\}$.
The former is used to set $\CC = \{1,\ldots,\left \lfloor{cn}\right \rfloor \}$.
The latter, which we refer to as \emph{membership-to-list probability}, is the probability that a color $j \in \CC$ belongs to $L(v)$, for each $v$.
Weights are set to 1 for all colors.
Random numbers are yielded by an uniform distribution.

The experiment consists of comparing two approaches, which we call CPLEX-CF and BP-SC. In the former, the ILP solver of CPLEX (with default parameters) solve the
MWLCP-CF formulation. The latter is our Branch-and-Price algorithm, which was manually implemented in \texttt{C++} and uses CPLEX only for the resolution of LP relaxations.
In order to speed up the generation of subproblems, we copy and perform minimum changes to the data structure where the instance lies.
This saves time since it is not necessary to re-compute the partition of the set of colors nor the subgraphs $G^k$ among them from scratch.
Lemma \ref{uncolor} is applied during this stage.
Then, subproblems are solved in a depth-first search fashion.
Regarding the pricing routine, we traverse graphs $G^k$ according to the value $T^k$ in decreasing order, in order to avoid solving the MWSSP over the same graph more than once.
On the other hand, our current implementation does not use the Hungarian Algorithm for the relaxations related to Corollary \ref{completos}.
This feature will be incorporated in a future version.

The experiment is carried out by a desktop computer equipped with an AMD Phenom II X4 3.4 GHz (a single thread is used), 3.6 GB of memory, Ubuntu 16.04
operating system, GCC 5.4.0 and IBM ILOG CPLEX 12.7 as the integer and linear programming solver.
A time limit of 1 hour is imposed on solving each instance.

Results are summarized in Table \ref{tab:my_label}.
In each row, averages over 5 instances generated with the same combination of values $(n, p, c, q)$ are reported:
averages of CPU time (in seconds) for CPLEX-CF and BP-SC with best times highlighted, and average of number of nodes explored by BP-SC.
Only instances solved to optimality within the time limit are considered in the averages.
For those cases when at least one of the 5 instances is not solved, the number of solved ones is reported in brackets.
If none of them is solved, a symbol ``--'' is displayed.

\begin{table} 
\footnotesize
\centering
\def\arraystretch{1.2}
\hspace{20pt}
\begin{tabularx}{\textwidth}{cccccccccccc}
&&&\multicolumn{3}{c}{$q$ = 0.25}&\multicolumn{3}{c}{$q$ = 0.50}&\multicolumn{3}{c}{$q$ = 0.75} \\ % \cmidrule(1r){4-6} \cmidrule(1r){7-9} \cmidrule(1r){10-12}
&&&\multicolumn{2}{c}{SC}&CF&\multicolumn{2}{c}{SC}&CF&\multicolumn{2}{c}{SC}&CF \\ %\cmidrule(1r){4-5} \cmidrule(1r){7-8} \cmidrule(1r){10-11}
$n$&$p$&$c$&nodes&time&time&nodes&time&time&nodes&time&time \\
\cmidrule{1-12}
\multirow{9}{*}{50}&\multirow{3}{*}{0.25}&0.5&47&0.3&0.3&233&\cellcolor{cellcolor}8.5&22.0&55&\cellcolor{cellcolor}2.9&1235.9 \\[-1pt]
&&1.0&475&5.4&\cellcolor{cellcolor}1.2&792&\cellcolor{cellcolor}33.6&660.3&53&\cellcolor{cellcolor}4.2&-- \\[-1pt]
&&1.5&88&\cellcolor{cellcolor}1.3&4.2&87&\cellcolor{cellcolor}3.3&454.8&56&\cellcolor{cellcolor}5.2&-- \\ %\cmidrule{2-12}
&\multirow{3}{*}{0.50}&0.5&46&\cellcolor{cellcolor}0.4&0.6&159&\cellcolor{cellcolor}4.2&400.9&50&\cellcolor{cellcolor}2.1&-- \\[-1pt]
&&1.0&57&\cellcolor{cellcolor}0.7&2.5&238&\cellcolor{cellcolor}8.1&507.8&45&\cellcolor{cellcolor}2.9&-- \\[-1pt]
&&1.5&211&\cellcolor{cellcolor}3.6&6.2&63&\cellcolor{cellcolor}2.5&623.4&37&\cellcolor{cellcolor}3.5&-- \\ %\cmidrule{2-12}
&\multirow{3}{*}{0.75}&0.5&39&\cellcolor{cellcolor}0.3&0.5&26&\cellcolor{cellcolor}0.5&27.0&16&\cellcolor{cellcolor}0.6&-- \\[-1pt]
&&1.0&22&\cellcolor{cellcolor}0.2&0.4&23&\cellcolor{cellcolor}0.6&53.1&13&\cellcolor{cellcolor}0.8&-- \\[-1pt]
&&1.5&30&\cellcolor{cellcolor}0.4&0.8&19&\cellcolor{cellcolor}0.7&64.1&15&\cellcolor{cellcolor}1.3&-- \\
\cmidrule{1-12}
\multirow{9}{*}{60}&\multirow{3}{*}{0.25}&0.5&107&1.4&\cellcolor{cellcolor}0.8&230&\cellcolor{cellcolor}8.6&106.0&895&\cellcolor{cellcolor}143.7&-- \\[-1pt]
&&1.0&627&14.1&\cellcolor{cellcolor}8.0&95&\cellcolor{cellcolor}7.1&789.9&184&\cellcolor{cellcolor}28.6&-- \\[-1pt]
&&1.5&261&\cellcolor{cellcolor}8.7&30.1&631&\cellcolor{cellcolor}102.3&-&726&\cellcolor{cellcolor}177.6&-- \\ %\cmidrule{2-12}
&\multirow{3}{*}{0.50}&0.5&43&\cellcolor{cellcolor}0.7&1.6&132&\cellcolor{cellcolor}4.9&596.5(2)&487&\cellcolor{cellcolor}49.2&-- \\[-1pt]
&&1.0&249&\cellcolor{cellcolor}5.5&25.9&71&\cellcolor{cellcolor}4.1&-&457&\cellcolor{cellcolor}68.6&-- \\[-1pt]
&&1.5&113&\cellcolor{cellcolor}3.4&34.3&121&\cellcolor{cellcolor}10.9&-&408&\cellcolor{cellcolor}86.1&-- \\ %\cmidrule{2-12}
&\multirow{3}{*}{0.75}&0.5&51&\cellcolor{cellcolor}0.6&0.9&33&\cellcolor{cellcolor}1.3&1686.6(3) &35&\cellcolor{cellcolor}2.3&-- \\[-1pt]
&&1.0&43&\cellcolor{cellcolor}0.8&3.1&34&\cellcolor{cellcolor}1.8&-&40&\cellcolor{cellcolor}3.9&-- \\[-1pt]
&&1.5&32&\cellcolor{cellcolor}0.8&2.9&45&\cellcolor{cellcolor}2.5&-&33&\cellcolor{cellcolor}4.7&-- \\ 
\cmidrule{1-12}
\multirow{9}{*}{70}&\multirow{3}{*}{0.25}&0.5&8467&223.6&\cellcolor{cellcolor}4.6&388&\cellcolor{cellcolor}55.8&1154.8(3)&71&\cellcolor{cellcolor}19.1&--\\[-1pt]
&&1.0&2669&\cellcolor{cellcolor}68.5&88.7&3167&\cellcolor{cellcolor}590.0(2)&--&95&\cellcolor{cellcolor}32.6&--\\[-1pt]
&&1.5&4814&\cellcolor{cellcolor}281.7(3)&894.0&1853&\cellcolor{cellcolor}230.5(3)&--&101&\cellcolor{cellcolor}43.0&--\\ %\cmidrule{2-12}
&\multirow{3}{*}{0.50}&0.5&523&\cellcolor{cellcolor}14.7&82.3&2333&\cellcolor{cellcolor}188.5&--&166&\cellcolor{cellcolor}34.1&--\\[-1pt]
&&1.0&1506&\cellcolor{cellcolor}59.3&632.6(4)&200&\cellcolor{cellcolor}14.9&--&153&\cellcolor{cellcolor}50.4&--\\[-1pt]
&&1.5&661&\cellcolor{cellcolor}26.1&1703.6&1506&\cellcolor{cellcolor}11.0&--&161&\cellcolor{cellcolor}71.3&--\\ %\cmidrule{2-12}
&\multirow{3}{*}{0.75}&0.5&59&\cellcolor{cellcolor}1.2&4.3&68&\cellcolor{cellcolor}3.2&--&45&\cellcolor{cellcolor}5.0&--\\[-1pt]
&&1.0&64&\cellcolor{cellcolor}1.7&30.7&40&\cellcolor{cellcolor}3.3&--&45&\cellcolor{cellcolor}8.3&--\\[-1pt]
&&1.5&77&\cellcolor{cellcolor}2.6&94.5&52&\cellcolor{cellcolor}4.9&--&46&\cellcolor{cellcolor}11.8&--
\end{tabularx}
\caption{Results on random graphs (average of 5 instances per row).}
\label{tab:my_label}
\end{table}

As we can see from the table, BP-SC achieves a remarkable performance, being able to solve to optimality 98\% of the instances (397 of 405).
In particular, all instances with 50 and 60 vertices can be solved in a couple of minutes.
The same behavior is observed when increasing the number of vertices to 70 and considering an edge probability of 0.5 or 0.75 (medium to high density).
In contrast, instances with edge probabilities of 0.25 (low density) are occasionally harder to solve, except when membership-to-list probability is high.

Note that the number of explored nodes is low, pointing out that the linear relaxations provide tight bounds and the node selection strategy is enough good,
but also the resolution of each relaxation is very time consuming.

Regarding CPLEX-CF, we notice that it outperforms BP-SC in barely 5\% of the cases.
These cases usually have low density graphs and lists with a small amount of colors, and makes MWLCP-CF become smaller in number of variables and constraints.
On the other hand, the lower the density of the graph, the harder the resolution of the MWSSP on graphs $G^k$ (whose tend to be sparse).

Last but not least, both approaches evidence perturbations in the performance not only when varying number of vertices and density but also when changing the structure
of the lists of colors. In the case of BP-SC, this behavior seems to be unpredictable. Even when we fix the number of vertices, edge probability and membership-to-list probability, the performance does not always get worse as the number of colors increases.

\section{Conclusions} \label{conclu}

In this work, we present a Branch-and-Price algorithm to solve the weighted version of the List Coloring Problem, which has many applications and generalizes several
other coloring problems.
In order to achieve that, we propose an ILP formulation, a column generation process and a robust branching scheme based on the ones given in \cite{trick1996} for the GCP. 
In particular, when restricting to instances coming from GCP, our approach behaves just like LPCOLOR.
However, other aspects should have been taken into account, such as the feasibility and the lack of indistinguishability between the colors.
As far as we are concerned, this is the first exact algorithm based on ILP for this problem.

The computational experiments show that our algorithm has a remarkable performance, being able to solve to optimality randomly generated instances up to 70 vertices and
within a time limit of 1 hour. Our results also expose that the hardness of the instances does not seem to depend only on the generation parameters ($n$, $p$, $c$ and $q$),
but also on the distribution of colors present in the lists.
In this respect, we believe that further research still remains to be done.
Particularly, it should be studied how the structure of graphs $G^k$, i.e.~size, density and overlap (number of vertices they share each other),
influence the performance of the Branch-and-Price algorithm.
A better understanding of tough instances can help to detect improvement opportunities in our algorithm. 

In future works, we plan to evaluate the performance of our algorithm over other types of instances. For example, by adding different weights to colors, varying
the number of indistinguishable colors, using instances coming from applications (see e.g.~\cite{ejs2018}) and other coloring problems (e.g.~Precoloring Extension,
$(\gamma,\mu)$-coloring) or by creating instances from benchmark graphs such as the ones from COLORLIB (DIMACS).

\end{document}